\documentclass{amsart}

\newcommand{\Ww}{\mathcal W}

\newcommand{\Ll}{\mathcal L}

\newcommand{\red}[1]{{\color{red}#1}}

\newcommand{\set}[1]{\{#1\}}

\newcommand{\nat}{\mathbb{N}}

\let\oldfrac\frac
\renewcommand{\frac}[2]{%
  \mathchoice
    {\oldfrac{#1}{#2}}
    {#1/#2}
    {#1/#2}
    {#1/#2}
}

\newcommand{\mso}{{\sc mso}\xspace}
\newcommand{\wmso}{{\sc wmso}\xspace}
\newcommand{\msoplus}{{\sc mso+}}
\newcommand{\msou}{{\sc mso+u}\xspace}

\newcommand{\eqdef}{\stackrel{\text{def}}=}

\theoremstyle{plain}
\newtheorem*{rep@theorem}{\rep@title}
\newcommand{\newreptheorem}[2]{%
\newenvironment{rep#1}[1]{%
\def\rep@title{#2 \ref{##1}}%
\begin{rep@theorem}}%
{\end{rep@theorem}}}

\newreptheorem{theorem}{Theorem}
\newreptheorem{lemma}{Lemma}
\newreptheorem{proposition}{Proposition}
\newreptheorem{conjecture}{Conjecture}

\usepackage{amsmath,xspace,tikz,centernot,color}
\usepackage[foot]{amsaddr}
\usetikzlibrary{matrix}

\newtheorem{theorem}{Theorem}[section]
\newtheorem{lemma}[theorem]{Lemma}
\newtheorem{corollary}[theorem]{Corollary}

\theoremstyle{definition}
\newtheorem{definition}[theorem]{Definition}
\newtheorem{example}[theorem]{Example}

\begin{document}
\title{Extensions of $\omega$-Regular Languages}         

\author{Miko{\l}aj Boja\'{n}czyk}
\address{University of Warsaw, Poland}

\author{Edon Kelmendi}
\address{University of Oxford, UK}

\author{Rafa{\l} Stefa\'{n}ski}
\address{University of Warsaw, Poland}

\author{Georg Zetzsche}
\address{Max Planck Institute for Software Systems (MPI-SWS), Germany}

\begin{abstract}
  We consider extensions of monadic second order logic over $\omega$-words, which are obtained by adding one language that is not $\omega$-regular. 
   We show that if the added language $L$ has a neutral letter, then the resulting logic is necessarily undecidable. A corollary is that the $\omega$-regular languages are the only decidable Boolean-closed full trio over $\omega$-words. 
\end{abstract}

\maketitle

\section{Introduction}
\label{sec:Introduction}
A famous theorem of B\"uchi~\cite[Theorem 2]{Buchi62} says that the monadic second-order theory of $(\omega,<)$ is decidable. What can  be added to this logic while retaining decidability?  This question has seen a lot of interest, and we begin by discussing some of the existing results. 

\paragraph*{What predicates can be added?} The first natural idea is to add predicates beyond the order $<$, e.g.~a unary predicate for the primes, or a binary addition function. 
This idea was pursued  already by Robinson in~\cite{Robinson1958}, in what is possibly the first published paper  to mention \mso on $(\omega,<)$. This is before B\"uchi's theorem about decidability of \mso, and even before the decidability results about weak \mso of B\"uchi~\cite[Corollary 1]{Buchi60}, Elgot~\cite[Corollary 5.8]{Elgot61} and Trakhtenbrot~\cite{trakthenbrot}. After  describing \mso,  which he credits to Tarski's lectures, Robinson    shows that adding the doubling function $n \mapsto 2n$ to \mso results in an undecidable logic~\cite[p.242]{Robinson1958}. Other examples of unary functions that lead to undecidability were given by  Elgot and Rabin~\cite[Section 1]{ElgotRabin66}. One of these examples is that  \mso becomes  undecidable after adding any function $f$ such that $f^{-1}(n)$ is infinite for all $n$. This result was strengthened by Siefkes~\cite[Theorem 5]{siefkes1971undecidable} who showed that it is enough that $f^{-1}(n)$ is infinite for all $n$ with certain periodicity properties, and then by Thomas~\cite[Theorem 1]{thomas1975note} who showed that it is enough for $f^{-1}(n)$ to be infinite for infinitely many $n$.  Another example of undecidability is \mso extended with any unary function  $f$ that is monotone and satisfies $f(n+1)>f(n)+1$ for infinitely many $n$, see~\cite[Theorem 2]{thomas1975note}. This line of research is summarised in~\cite{RabinovichThomas2006} as follows: ``for most examples of natural functions or binary relations it turned out that the corresponding monadic theory is undecidable, usually shown via an interpretation of first-order arithmetic''.

The undecidability issues mentioned above are avoided if one considers unary predicates.  The first examples of this kind were given by Elgot and Rabin, who showed that \mso remains decidable after adding unary predicates for the factorials, or  squares, or cubes, etc.~see~\cite[Theorem 4]{ElgotRabin66}). 
Following this result,  a lot of attention has been devoted to identifying the  unary predicates that keep \mso decidable. An equivalent phrasing of this question is: which $\omega$-words have a decidable \mso theory?
 An interesting example is the Thue-Morse word;  its \mso theory is decidable, which follows from~\cite[Theorem 3]{muchnik03_almos_period_sequen}.  A general classification of $\omega$-words with a decidable \mso theory was given by Semenov in~\cite[p.~165]{Semenov84b}, this line of research was continued in ~\cite{carton2002monadic,  RabinovichThomas2006}. It is worth pointing out that the classification can be hard to apply to some specific cases; an important one being the case of  prime numbers. It is unknown if \mso extended with a predicate for the prime numbers has a decidable theory; if this were the case  then one could use the algorithm to decide if there are infinitely many twin primes\footnote{This theory is known to be decidable if one assumes  Schinzel's Hypothesis --  a conjecture  from number theory which  implies that there are infinitely many  twin primes~\cite[Theorem 4]{Bateman1993}.}.

 It is worth pointing out that in all of the results discussed above, it makes no difference whether one uses \mso or weak \mso. The undecidability proofs for unary functions in~\cite{ElgotRabin66,thomas1975note,siefkes1971undecidable} use only weak \mso. For the results about unary predicates, it makes no difference if \mso or weak \mso is used, because for every $\omega$-word, its \mso theory is decidable if and only if its weak \mso theory is decidable, which follows from  McNaughton's determinisation theorem~\cite[p.~524]{McNaughton66}\footnote{The equivalence of \mso and \wmso need not hold after  adding non-unary predicates. For example, \wmso with addition can only define languages in the Borel hierarchy, while \mso with addition can easily be shown to contain  the logic \msou that will be discussed later in the paper, and \msou can define languages beyond the Borel hierarchy~\cite[Theorem 2.1]{hummel2012topological}.}. In a sense, one could say that the results discussed above are really about extending weak \mso. This will no longer be true when adding quantifiers and languages.

 \paragraph*{What quantifiers can be added?} Another line of research concerns adding new quantifiers. 
 If the added quantifier has some implicit arithemetic, such as the H\"artig quantifier~\cite{herre1991hartig}, which expresses the existence of two sets of equal size with a given property, then \mso immediately becomes undecidable. This follows directly from Robinson's result about $n \mapsto 2n$, and it is also discussed in more detail in~\cite[Theorem 14]{klaedtkerues2003}. However, there are quantifiers which describe only topological or asymptotic behaviour, and for such quantifiers proving undecidability  can be much harder. 
 One example of an asymptotic  quantifier is the bounding quantifier from~\cite{bojanczykBoundingQuantifier2004}, which expresses the property ``$\varphi(X)$ is true for finite sets $X$ of unbounded size''~\cite{bojanczykBoundingQuantifier2004}. The resulting logic, called \msou, is undecidable~\cite[Theorem 1.1]{bojanczyk2016}. However, it is close to the decidability border; in particular weak \mso with the bounding quantifier is decidable~\cite[Theorems 3 and 5]{bojanczykWeakMSOUnbounding2011}; and the same is true for its  variants and extensions~\cite[Theorems 11 and 13]{bojanczykDeterministicAutomataExtensions2009}. Other extensions of \mso with asymptotic quantifiers were proposed by Michalewski, Mio and Skrzypczak in~\cite{michalewski2016measure,mioskmi2018}, including a quantifier related to  Baire category and a quantifier related to  probability. The Baire quantifier does not add to the expressive power of \mso~\cite[Theorem 4.1]{mioskmi2018}. On the other hand, the  probability quantifier leads to an undecidable logic~\cite[Theorem 1]{michalewski2016measure}, because it can express the undecidable problem of checking if a probabilistic B\"uchi automaton accepts some word with nonzero probability~\cite[Theorem 7.2]{BaierBertrandGrosser2012}. The theme for quantifiers seems to be that adding a well-behaved quantifier to (non-weak) \mso either does not change the expressive power, or leads to an undecidable logic; but in the latter case the undecidability proof can be hard.

 \paragraph*{What languages can be added?} We now turn to the final kind of feature that can be added to \mso, namely languages. This is the main topic of this paper. A language $L \subseteq \set{a,b}^\omega$ over a binary alphabet can be viewed as a second-order unary predicate $L(X)$, which inputs a set $X \subseteq \omega$ and returns true if the language $L$ contains the word where positions from $X$ have label $a$ and the remaining positions have label $b$. This can be generalized to alphabets with $k$ letters: Then, the predicate inputs $k$ sets and returns true if the $k$ sets form a partition of $\omega$ and $L$ contains the word encoding this partition. Let us write \msoplus$L$ for the extension of \mso which has the predicate described above.
 
 Another equivalent way of describing the logic \msoplus$L$ uses closure properties of languages. A folklore fact about \mso is that  existential monadic quantification is the same as taking the image of a language under a letter-to-letter homomorphism, see~\cite[p.~2]{Rabin69} or~\cite[Section 2.3]{thomasLanguagesAutomataLogic1997}.  It follows that  \msoplus$L$ is exactly  the smallest class of languages of $\omega$-words which contains $L$ and  all $\omega$-regular languages, and which is closed under (a) Boolean combinations; (b)
 images of letter-to-letter homomorphisms; and (c) inverse images of letter-to-letter homomorphisms.  We will return to this language theoretic approach in Section~\ref{sec:cones}.

 \begin{example}\label{ex:singleton}
  Suppose that $L \subseteq \Sigma^\omega$ is a singleton language, i.e.~it contains only one word $w$. Then \msoplus$L$ has  the same expressive power as \mso extended with  unary predicates  
  \begin{align*}
  P_a = \set{n \in \omega : \text{the $n$-th letter of $w$ is $a$} } \quad \text{for $a \in \Sigma$.}
  \end{align*}
  Therefore, in the case of singleton languages \msoplus$L$, corresponds to the \mso extensions with unary predicates that were studied in~\cite{ElgotRabin66,thomas1975note,Semenov84b,carton2002monadic,RabinovichThomas2006}.
\end{example}

\begin{example}\label{ex:u-language}
  Define $U \subseteq \set{a,b}^\omega$ to be the $\omega$-words where blocks  of $a$'s have unbounded size:
  \begin{align*}
  U \eqdef \set{a^{k_1}ba^{k_2} \cdots : \limsup k_n = \infty}.
  \end{align*}
  In~\cite[Theorem 1.3]{bojanczykUndecidabilityMSOUltimately2018} it is shown that adding the language $U$ to \mso gives exactly the logic \msou. Hence it is unambiguous to write \msou, with both meanings (adding a quantifier or a language) being equivalent. As mentioned before, this logic is undecidable.
\end{example}
 
\begin{example}\label{ex:ultimately-periodic}
  Consider the   ultimately periodic words, i.e.
  \begin{align*}
  P = \set{wv^\omega : w,v \in \set{a,b}^* \text{ and $v$ is nonempty}}.
  \end{align*}
  It is not hard to see that \msoplus$P$ can express the language $U$ from Example~\ref{ex:u-language}, and therefore this logic is undecidable~\cite[Theorem 1.4]{bojanczykUndecidabilityMSOUltimately2018}. In contrast, adding $P$ to weak \mso yields a decidable logic~\cite[Theorem 13]{bojanczykDeterministicAutomataExtensions2009}.
\end{example}

Our goal in this paper is to classify the languages $L$ such that \msoplus$L$ is undecidable. By the discussion in  Example~\ref{ex:singleton}, this project is at least as difficult as classifying the $\omega$-words with a decidable \mso theory. However, in the spirit of ``asymptotic'' conditions, we restrict attention to languages which have a neutral letter, which means that there is a letter in the alphabet, denoted by $1$,  such that 
\begin{align*}
w_1 1 w_2 1 \cdots \in L \qquad\Leftrightarrow\qquad w_1 w_2 \cdots \in L
\end{align*}
holds for every words $w_1,w_2,\ldots \in \Sigma^*$ where infinitely many $w_i$ are nonempty\footnote{In the  definition of neutral letters, we require that the language is stable under  inserting or deleting infinitely many neutral letters. However, this also implies that the language is stable under inserting or deleting finitely many neutral letters.}. We now state the main theorem of this paper. 
\newcommand{\thmmso}{ If   $L \subseteq \Sigma^\omega$ has a neutral letter and is not definable in \mso, then \msoplus$L$ is undecidable.}
\begin{theorem}
  \label{th:mso el}
  \thmmso
\end{theorem}
In the above theorem, by undecidable we mean that there is no algorithm which decides the sentences of \msoplus$L$ that are true in $(\omega,<)$. Since \mso can quantify over words, this is the same as saying that satisfiability is undecidable for \msoplus$L$ for $\omega$-words. The proof of Theorem~\ref{th:mso el} will be given in Section~\ref{sec:congruence game}.

\begin{example}
  \label{ex:u-neutral}
  Define $U' \subseteq \set{a,b,1}^\omega$ to be the words such that eliminating all $1$'s gives a word in the language $U$  from Example~\ref{ex:u-language}. We claim that the logic is obtained by extending \mso with (a) the bounding quantifier; or (b) the language $U$; or (c) the language $U'$. The equality of (a) and (b) was discussed in Example~\ref{ex:u-language}. The language $U'$ can be defined using the bounding quantifier, hence the inclusion (c) $\subseteq$ (a). The language $U$ is the intersection of $U'$ with language of words that do not contain the letter $1$, hence the inclusion (b) $\subseteq$ (c). The equality of these three logics is discussed in more detail in~\cite{bojanczykUndecidabilityMSOUltimately2018}.
\end{example}

The language $U$ will play an important role in the proof of Theorem~\ref{th:mso el}. We will show that if $L$ is not definable in \mso and contains a neutral letter, then $U$ is definable in \msoplus$L$. Undecidability will then follow by Theorems~\cite[Theorem 1.1]{bojanczyk2016} and~\cite[Theorem 1.3]{bojanczykUndecidabilityMSOUltimately2018}. In this sense, $U$ is the simplest undecidable extension of \mso.

The paper is structured as follows. In Section~\ref{sec:cones}, we discuss a version of our main theorem for finite words, which was proved by Zetzsche et al.~in~\cite{zetzsche16_boolean_closed_full_trios_ration_kripk_frames}. Like~\cite{zetzsche16_boolean_closed_full_trios_ration_kripk_frames}, we prove our main theorem using  syntactic congruences, and therefore  Section~\ref{sec:congruences} is devoted to a discussion of syntactic congruences for $\omega$-languages. In Section~\ref{sec:congruence game}, we prove our main result, and in Section~\ref{sec:full trios} we show that the main theorem implies that the $\omega$-regular languages are the only Boolean-closed full trios that are decidable. 

\section{Finite words}
\label{sec:cones}
In this section, we describe the starting point for our work, which is a theorem by Zetzsche et al., which says that the regular languages of finite words are the only decidable Boolean-closed full trio. To define full trios\footnote{Full trios are sometimes called {\em cones} in formal languages literature, they are meant to be a formalisation of robust classes of languages.}, recall that a \emph{homomorphism}  is a function
\begin{align*}
h : \Sigma^* \to \Gamma^*  \qquad \text{such that }h(wv)=h(w)h(v).
\end{align*}
Define the \emph{arithmetic hierarchy}, see~\cite[Section 2]{zetzsche16_boolean_closed_full_trios_ration_kripk_frames}, to be the least class of languages of finite words that contains all recursively enumerable languages, and which is closed under complementation and homomorphic images. 
\begin{theorem}\label{thm:finite}\cite[Corollary 3.2]{zetzsche16_boolean_closed_full_trios_ration_kripk_frames} Let $\Ll$ be a class of languages of finite words which is a full trio, i.e.~it is closed under:
  \begin{enumerate}
    \item images under homomorphisms; and
    \item inverse images under homomorphisms; and
    \item intersections with  regular languages. 
    \end{enumerate}
    If $\Ll$ is additionally Boolean-closed (closed  under union and complementation) and it  contains at least one non-regular language, then it contains the arithmetic hierarchy.
\end{theorem}

In the above theorem, there is no assumption on neutral letters. This is because  condition 2 deprecates the assumption,  since a neutral letter can be added to any language by taking the inverse image under the homomorphism which eliminates the neutral letter. The closure properties used in Theorem~\ref{th:mso el} are weaker, and hence the assumption on neutral letters is needed.
As a warm-up for the case of $\omega$-words, we give below a proof sketch for the above theorem. 
\begin{proof}[Proof sketch]
  The  proof uses rational relations~\cite[p.236]{Eilenberg74}.  Recall that a rational relation is a binary relation on words 
  that is recognised by a nondeterministic automaton where each transition is labelled by a pair (input word, output word), with both words being possibly empty. By Nivat's theorem    (Propositions 1 and 2 in \cite{nivat}), if a language class $\Ll$ is a full trio, then it is closed under images under rational relations, which can be visualised as the following reasoning rule:
  \begin{align*}
    \frac{K \subseteq \Sigma^* \text{ is in $\Ll$} \qquad R \subseteq \Sigma^* \times \Gamma^* \text{ is a rational relation} }{\set{v \in \Gamma^* : \exists w \in L \text{ with }(w,v) \in R} \text{ is in $\Ll$}}
    \end{align*}
    Because $\Ll$ is closed under complementation, we can also use a variant of the above rule where $\forall$ is used instead of $\exists$ in the conclusion of the rule (i.e.~below the line). 
  
  The key idea is to use the closure properties to formalise the syntactic right congruence of the language. Let $L \subseteq \Sigma^*$  be some non-regular language in $\Ll$, which exists by assumption, and  let $\sim$ be its syntactic right congruence, i.e.~the equivalence relation  defined by
  \begin{align*}
  u \sim u' \quad \eqdef \quad \forall v \in \Sigma^*\ \ \   uv \in L \Leftrightarrow u'v \in L.
  \end{align*}
  By the Myhill-Nerode theorem, $\sim$ has infinite index, i.e.~infinitely many equivalence classes.  
  Using the reasoning rule with rational relations, one  shows that $\Ll$ contains the language 
  \begin{align*}
  L_1 = \set{u\# u' : u \sim u'},
  \end{align*}
  where  $\#$ is a fresh separator symbol.
  Consider now two separator symbols $\#$ and $\red{\#}$. Define $L_2$ to be the language 
  \begin{align*}
  \set{w_1\# \cdots \# w_n  \# v_1\red{\#} \cdots \red{\#} v_m \red{\#} : \begin{cases}
    w_1 \sim v_1 \\
    w_n \sim v_m \\
    w_i \not \sim w_j \text{ for $i \neq j$}\\
    v_i \not \sim v_j \text{ for $i \neq j$}\\
    w_i \sim v_j \Rightarrow w_{i+1} \sim v_{j+1}
  \end{cases}}.
\end{align*} Using the closure properties, one shows $L_2 \in \Ll$.
A short analysis of the conditions defining $L_2$ reveals that every word in $L_2$ satisfies
\begin{align*}
m=n \qquad \text{and} \qquad w_1 \sim v_1, w_2 \sim v_2,\ldots,w_n\sim v_n.
\end{align*}
Furthermore, since  $\sim$ has infinitely many equivalence classes, it follows that $n=m$ can be arbitrarily large.   By projecting away the words $w_i, v_i$ using a homomorphisms, it follows  that $\Ll$ contains the language
\begin{align*}
L_3 = \set{\#^n \red{\#}^n : n \in \set{1,2,\ldots}}.
\end{align*}
A string encoding of runs of two-counter machines, see~\cite[Theorem 2]{DBLP:journals/jcss/HartmanisH70}, can be used to show that  $\Ll$ contains every recursively enumerable language. The arithmetic hierarchy follows, by closure of $\Ll$ under homomorphic images and complementation.
\end{proof}

\section{Congruences for $\omega$-words}
\label{sec:congruences}
Like in Theorem~\ref{thm:finite},  the proof of Theorem~\ref{th:mso el} also uses congruences. However, there are several issues with congruences for $\omega$-words, which mean that some new ideas are needed. The main problem is that there is no good notion of syntactic congruence for  languages of $\omega$-words.  

 We begin by discussing several existing approaches to congruences for $\omega$-words, see also~\cite{MalerStaiger93}. In all cases, we begin with a language $L \subseteq \Sigma^\omega$, and use it to define an equivalence relation on finite words. The first candidate is  the \emph{right congruence}, which identifies two finite words $u,u' \in \Sigma^*$ if 
\begin{align*}
\forall v \in \Sigma^\omega\qquad \ uv \in L \Leftrightarrow u'v \in L.
\end{align*}
This right congruence does not characterise the $\omega$-regular languages, although it does have some use, for example in automata learning~\cite{DBLP:journals/corr/abs-1809-03108}. There could be finitely many equivalence classes despite a language not being $\omega$-regular. 
For example, every prefix independent language will have one equivalence class of right congruence, but there are prefix independent languages which are not $\omega$-regular, such as
\begin{align*}
\set{w (a^nb)^\omega : w \in \set{a,b}^* \text{ and $n$ is prime}}.
\end{align*}
By induction on formula size one can show that if $L$ has a right congruence of finite index, then the same is true for every language definable in \msoplus$L$; which shows that right congruences will not be useful for our main result. Similar problems arise for the two-sided version of right congruence. 

A more useful congruence for $\omega$-words uses  two-sided environments and $\omega$-iteration; this leads to the  \emph{Arnold congruence}~\cite[Section 2]{Arnold85b}, which identifies $u,u'$ if 
\begin{align*}
   \land \begin{cases} 
    \forall w \in \Sigma^*\ \forall v \in \Sigma^\omega \qquad wuv \in L \Leftrightarrow wu'v \in L\\
      \forall w,v \in \Sigma^*\qquad w(uv)^\omega \in L \Leftrightarrow w(u'v)^\omega \in L.
   \end{cases}
  \end{align*}
The Arnold congruence still does not characterise the $\omega$-regular languages. For example, the language $U$ is not $\omega$-regular, but it has two equivalence classes under Arnold congruence: words which contain $b$, and words which do not contain $b$.

Fortunately, there is a successful characterisation of $\omega$-regular languages via congruences. This characterisation is stated below, and it corresponds to $\omega$-semigroups.
\begin{theorem}\cite[Theorem 7.5]{DBLP:books/daglib/0016866}\label{thm:omega-congruence}
  A language $L \subseteq \Sigma^\omega$ is $\omega$-regular if and only if there is an equivalence relation $\sim$ which has finite index and satisfies the following conditions for all sequences of finite words $u_i$:
  \begin{eqnarray}
    \label{eq:concatenation-congruence}
   \bigg(\bigwedge_{i\in\set{1,2} } u_i \sim u_i'\bigg) & \Rightarrow &  u_1 u_2 \sim u'_1 u'_2
\\\label{eq:recognising-congruence}
   \bigg(\bigwedge_{i\in\nat } u_i \sim u_i'\bigg)\ \  & \Rightarrow &  \left(u_1 u_2 \cdots \in L \Leftrightarrow u'_1 u'_2 \cdots \in L\right)
    \end{eqnarray}
\end{theorem}
We use the name \emph{$\omega$-congruence} for an equivalence relation that satisfies conditions~\eqref{eq:concatenation-congruence} and~\eqref{eq:recognising-congruence} in the above theorem. 
We use the above theorem in our main result. Although promising, the  characterisation in terms of $\omega$-congruences has one important drawback, namely non-uniqueness. For the right congruence, the defining property
\begin{align*}
  \forall v \in \Sigma^\omega\qquad uv \in L \Leftrightarrow u'v \in L.
  \end{align*}
gives a unique equivalence relation, and  hence it makes sense to speak of \emph{the} right congruence. A similar property holds for the Arnold congruence. The uniqueness of the definition of right congruence, and the fact that its definition can be formalised using rational relations, is what drives the proof of Theorem~\ref{thm:finite}. 

In contrast, there is no uniqueness in Theorem~\ref{thm:omega-congruence}, and there cannot be. A language  might not have a unique coarsest $\omega$-congruence (such an equivalence relation is called the \emph{syntactic $\omega$-congruence}). An example is the language $U$, see~\cite[Running Example 2]{bojanczykRecognisableLanguagesMonads2015}. For  $\omega$-regular languages, the syntactic $\omega$-congruence exists and  coincides with  Arnold congruence, see~\cite[Proposition 8.8]{DBLP:books/daglib/0016866}, but  this is not very helpful in our setting, since we want to study congruences for languages that are \emph{not} $\omega$-regular. These are issues that we will need to overcome in the proof of our main result. 

We finish this section  with  a simple observation, which says that condition~\eqref{eq:concatenation-congruence} in Theorem~\ref{thm:omega-congruence} is superfluous. This observation will be useful later on, since condition~\eqref{eq:recognising-congruence} will be easier to formalise. 
   
\begin{lemma}\label{lem:congruence}
  If there is an equivalence relation of finite index which satisfies~\eqref{eq:recognising-congruence}, then there is an equivalence relation of finite index which satisfies both~\eqref{eq:concatenation-congruence} and ~\eqref{eq:recognising-congruence}.
\end{lemma}
\begin{proof} Induction on the number of equivalence classes in the equivalence relation, call it~$\sim$. In the base case, when $\sim$ has one equivalence class, condition~\eqref{eq:concatenation-congruence} holds vacuously. 
    Consider the induction step. For this proof, it is easier to work with the following equivalent form of~\eqref{eq:concatenation-congruence}:
    \begin{align*}
      u \sim u'  \qquad \Rightarrow \qquad  (uw \sim u'w) \land (wu \sim wu').
    \end{align*}
    If $\sim$ satisfies the above implication, then we are already done. Otherwise,  choose a violation of the implication, i.e.~words $u \sim u'$ which do not  satisfy the conclusion of the implication. 
    By symmetry, assume $uw \not \sim  u'w$. Define  $\approx$ to be the equivalence relation obtained from $\sim$ by merging the equivalence classes of $uw$ and $u'w$. We claim that $\approx$ still satisfies condition~\eqref{eq:recognising-congruence}, and therefore the induction assumption can be applied. 
We visualize~\eqref{eq:recognising-congruence} as  follows:
  \vspace{-0.4cm}
  \begin{center} 
  \begin{tikzpicture}
    \matrix (m) [matrix of math nodes,row sep=1.5em,column sep=1em]
    {
     u_0 & u_1 & u_2 & \cdots & \in L \\
     u'_0 & u'_1 & u'_2 & \cdots & \in L \\
   };
   \path[-stealth]
   (m-1-1) edge[draw=none] node[sloped] {$\approx$} (m-2-1)
   (m-1-2) edge[draw=none] node[sloped] {$\approx$} (m-2-2)
   (m-1-3) edge[draw=none] node[sloped] {$\approx$} (m-2-3)
   (m-1-4) edge[draw=none] node[pos=0.45, sloped] {$\approx$} (m-2-4)
   (m-1-5) edge[draw=none] node[sloped] {$\Leftrightarrow$} (m-2-5);
  \end{tikzpicture} 
\end{center}
\vspace{-0.5cm} By definition of $\approx$ and assumption~\eqref{eq:recognising-congruence} for $\sim$, we can replace
every $u_n$ in the equivalence class of $uw$  by $uw$ and every $u_n$ in the equivalence class of $u'w$  by $u'w$,  without affecting membership in $L$. Therefore we can assume without loss of generality that every $u_n$ is either $uw$, or $u'w$, or a word that is $\sim$-equivalent to neither of these. The same can be done for $u'_n$.  By definition of $\approx$, if $u_n=uw$ then $u'_n$ has to be one of $uw$ or $u'w$. We can now split each $u_n = uw$ into two words $u$ and $w$, likewise for $u_n=u'w$, and then  use again the assumption that $\sim$ satisfies~\eqref{eq:recognising-congruence} to finish the proof. 
\end{proof}

\section{Proof of the main theorem}
\label{sec:congruence game}

In this section we prove Theorem~\ref{th:mso el}. Fix a language $L \subseteq \Sigma^\omega$  that is not $\omega$-regular, and which contains a neutral letter. We will show that the logic \msoplus$L$ is undecidable. 

To prove undecidability, we will show that  \msoplus$L$ contains the language $U$, and therefore undecidability follows thanks to the results about the logic \msou. To explain how $U$ can be defined, we use a game called the \emph{congruence game}.   This game is  played by two players called  Spoiler and Duplicator, and it is parametrised by an $\omega$-word $u \in \set{a,b}^\omega$. The congruence game is designed so that player Duplicator wins if and only if $u \in U$, which means that $u$ has $a$-labelled intervals of unbounded size.  This is achieved as follows.   Roughly speaking, the goal of player Duplicator is to show that from the perspective of the language $L$,  each finite word $w \in \Sigma^*$ is equivalent  to some word $v \in \Sigma^*$ which can fit infinitely often into $a$-labelled intervals in the word $u$. Since the language $L$ is not $\omega$-regular, Duplicator needs  intervals of unbounded size to win.

Define an \emph{interval} to be a finite connected subset of $\omega$, i.e.~it contains all positions between its first and last position. If $W,V$ are intervals, then we write $W<V$ if the last position of $W$ is strictly before the first position of $V$.

\newcommand{\sequbis}[2]{\langle#1 \rangle_{#2}}
\begin{definition}
  [Congruence Game]  The \emph{congruence game for $u \in \set{a,b}^\omega$} is the following game played by two players, called Spoiler and Duplicator.
  \begin{enumerate}
    \item Spoiler chooses an infinite family $\Ww$ of pairwise disjoint    intervals.
\item Duplicator chooses intervals 
\begin{align*}
W_1 < V_1 < W_2 < V_2 < \cdots 
\end{align*}
such that $W_1,W_2,\ldots$ are from $\Ww$ and $V_1,V_2,\ldots$ contain only positions with label $a$ in the word $u$.
\item Spoiler chooses words
\begin{align*}
w_1,w_2,\ldots \in \Sigma^*
\end{align*}
such that $|w_i|<|W_i|$ for every $i \in \set{1,2,\ldots}$.
\item Duplicator chooses words
\begin{align*}
v_1,v_2,\ldots \in \Sigma^*
\end{align*}
such that $|v_i|<|V_i|$ for every $i \in \set{1,2,\ldots}$.
    \item  Spoiler chooses a sequence of natural numbers
    \begin{align*} 
     i_1 < i_2 < \cdots.
    \end{align*}
    \item Duplicator wins the game if and only if: 
       \begin{align*}
         w_{j_1}w_{j_2}w_{j_3}\cdots \in L\qquad \Leftrightarrow \qquad v_{j_1}v_{j_2}v_{j_3}\cdots \in L.
       \end{align*}
  \end{enumerate}
  
\end{definition}

The key result about the congruence game is the following lemma. The lemma does not use the assumption that $L$ contains a neutral letter; this assumption will be used later when formalising the congruence game in \msoplus$L$. 
\begin{lemma}\label{lem:congruence-game}
  Assume that $L \subseteq \Sigma^\omega$ is not $\omega$-regular.  Then
  \begin{align*}
  \text{Duplicator wins the congruence game for $u$}\ \  \Leftrightarrow\ \  u \in U.
  \end{align*}
   \end{lemma}
\begin{proof}\ 

($\Leftarrow$) Assume $u \in U$. We will show a winning strategy for player Duplicator. Suppose that player Spoiler has chosen a family $\Ww$ in round 1. Since $u \in U$, intervals with only $a$-labelled positions have unbounded size, and therefore in round 2, player Duplicator can choose the intervals so that  $|V_i| \ge |W_i|$ for all $i$. For every choice of words $w_i$ made by player Spoiler in round 3, Duplicator's response in round 4 is to choose the words $v_i$ so that  $v_i = w_i$ for all $i$.  This guarantees victory for Duplicator, regardless of Spoiler's move in round 5. 
  
  ($\Rightarrow$)  Assume $u \not \in U$. We will show a winning strategy for player Spoiler. In round 1, Spoiler picks $\Ww$ so that  the lengths of the intervals tend to infinity, i.e.~no size appears infinitely often.  Let $W_i$ and $V_i$ be the intervals that are chosen in round 2 by player Duplicator. By choice of $\Ww$, the lengths of the intervals $W_i$ tend to infinity, while by assumption that $u \not \in U$, the lengths of the intervals $V_i$ are bounded. In round 3, Spoiler chooses the  words $w_i$ so that every word from $\Sigma^*$ appears infinitely often. This  can be done because the lengths of the intervals $W_i$ tend to infinity. Suppose that Duplicator chooses some words $v_i$ in round 4. Since the intervals $V_i$ have bounded size,   the words $v_i$ chosen by Duplicator come from a finite set $F\subseteq\Sigma^*$.   This means that for every  $w \in \Sigma^*$, there is some $v\in \Sigma^*$ such that infinitely often $w_i=w$ and $v_i=v$. Choose some function  
    \begin{align*}
    f : \Sigma^* \to F
    \end{align*}
     which realises the dependency $w \mapsto v$, i.e.~for every $w \in \Sigma^*$, 
    \begin{align}\label{eq:io-response}
       w = w_i \quad \text{and} \quad f(w)=v_i \qquad \text{for infinitely many $i$.}
    \end{align}
    Apply Lemma~\ref{lem:congruence} with $\sim$ being the kernel of $f$, i.e.~the equivalence relation that identifies two words if they have the same image under $f$.  Since $L$ is not $\omega$-regular, then by Lemma~\ref{lem:congruence} there must be a violation of~\eqref{eq:recognising-congruence}, i.e.~there must be words $u_1,u_2,\ldots$ such that 
\begin{align}\label{eq:non-equiv}
u_1 u_2 \cdots \in L \qquad  \not\!\!\!\!\!\iff\qquad  f(u_1) f(u_2) \cdots \in L.
\end{align}
By~\eqref{eq:io-response}, in Round 5,  Spoiler can choose the indices $i_1 < i_2 < \cdots$ so that 
\begin{align*}
u_n = w_{i_n} \qquad f(u_n) = f(v_{i_n}) \qquad \text{for all $n \in \set{1,2,\ldots}$}
\end{align*}
and therefore win thanks to~\eqref{eq:non-equiv}.
\end{proof}

The following corollary, and undecidability of the logic \msou, complete the proof of Theorem~\ref{th:mso el}.
\begin{corollary}\label{cor:u}
  If $L \subseteq \Sigma^\omega$ is not $\omega$-regular and contains a neutral letter, then \msou $\subseteq$  \msoplus$L$.
\end{corollary}
\begin{proof}
  By~\cite[Theorem 1.3]{bojanczykUndecidabilityMSOUltimately2018}, it is enough to show that the language $U$ is definable in \msoplus$L$. By Lemma~\ref{lem:congruence-game}, it is enough to show that \msoplus$L$ can express that Duplicator has a winning strategy in the congruence game. 
  
  A family of disjoint intervals is represented by two sets of positions: the set $X$ of leftmost positions in the intervals, and the set $Y$ of rightmost positions in the intervals. The condition that all intervals  are disjoint means that 
  \begin{align*}
  \forall x_1, x_2 \in X\ x_1 < x_2 \ \Rightarrow \   \exists y \in Y \ x_1 \le y < x_2.
  \end{align*}
  Using this representation, the choices of intervals in rounds 1 and 2 can be represented by set quantification, with choices of player Duplicator using existential quantifiers and choices of player Spoiler using universal quantifiers.  The words $w_i$ chosen in round 3 are represented by colouring the intervals $W_i$ with letters from $\Sigma$; and using the neutral letter for positions not in the intervals $W_i$.  The same goes for round 4. The subsequence in round 5 is represented by a subset of leftmost positions in the intervals $W_i$. The winning condition in round 6 is checked by using the predicate for $L$.  
\end{proof}

\section{Boolean closed full trios}
\label{sec:full trios}
We finish the paper with a corollary of our main theorem, which is an analog of Theorem~\ref{thm:finite} for infinite words: if a Boolean-closed full trio of languages of infinite words contains at least one non-regular language, then it contains the entire arithmetic hierarchy, subject to a certain representation.

As mentioned in the introduction, a language of $\omega$-words is definable in \msoplus$L$ if and only if it belongs to the smallest class of languages that contains $L$, contains all $\omega$-regular languages, is closed under Boolean combinations, as well as images and inverse images under letter-to-letter homomorphisms. If we lift the restriction on homomorphisms being letter-to-letter, then we get a Boolean-closed full trio, as discussed below.

When applying a homomorphism that may erase some letters, an $\omega$-word can be mapped to a finite word. Therefore, in the presence of such homomorphisms, it makes sense to consider languages of words of length $\le \omega$, i.e.~words which are either finite or $\omega$-words. For such words, define a \emph{regular language}  to be a union of two languages: a regular language of finite words, plus a regular language of $\omega$-words.
Define a \emph{homomorphism} for words of length $\le \omega$ to be a function
\begin{align*}
h: \Sigma^{\le \omega} \to \Gamma^{\le \omega} \qquad 
\end{align*}
which is obtained by applying to each letter a function of type $\Sigma \to \Gamma^*$. 

\begin{theorem}\label{thm:full trios}
  Let $\Ll$ be a class of languages of words of length at most $\le \omega$ which is closed under:
  \begin{enumerate}
    \item images under homomorphisms; and
    \item inverse images under homomorphisms; and
    \item intersections with regular languages. 
    \end{enumerate}
    If $\Ll$ is additionally Boolean-closed (closed  under union and complementation) and it  contains at least one non-regular language, then for every $L \subseteq \Sigma^*$ in the  arithmetic hierarchy,
    \begin{align*}
    \underbrace{\set{wv^\omega : w \in \Sigma^*,  v \in L}}_{\text{loop representation of $L$}} \in \Ll.
    \end{align*}
\end{theorem}
\begin{proof}
  By closure under inverse images of homomorphisms and under intersection with $\Sigma^*$ and $\Sigma^\omega$ for any $\Sigma$, if $\Ll$ contains some non-regular language, then it contains some non-regular $\omega$-language with a neutral letter.
  By Corollary~\ref{cor:u}, $\Ll$ contains all languages definable in  \msou. By~\cite[Lemma~3.2]{bojanczyk2016}, for every recursively enumerable language $L \subseteq \Sigma^*$,  the logic \msou   defines some $\omega$-language $K$, over an alphabet extended with a neutral letter,  such that 
  \begin{align*}
  \text{loop representation of $L$} = h(K),
  \end{align*}
  where $h$ is the homomorphism that eliminates the neutral letter.  Since $\Ll$ is closed under homomorphic images, it follows that $\Ll$ contains the loop representations of all recursive enumerable languages. For the arithmetic hierarchy, it is enough to observe that the class
  \begin{align*}
  \set{L \subseteq \Sigma^* : \text{$\Ll$ contains the loop representation of $L$}}
  \end{align*}
  is closed under Boolean combinations and homomorphic images.  
\end{proof}
In contrast to the finite word setting, we cannot conclude that $\Ll$
contains $L$ itself: Standard arguments show that if
a language $K\subseteq\Sigma^{\le\omega}$ has a finite-index right
congruence, then every language obtained from $K$ using Boolean full
trio operations also has a finite-index right congruence. Thus, for
example, if we start with $U$, all obtainable languages over finite
words are regular.

One could also consider Boolean closed full trios of
$\omega$-languages. Then, homomorphisms would be defined by functions
of type $\Sigma\to\Gamma^*$ and the (inverse) image of a homomorphism
on a subset of $\Sigma^\omega$ would contain only those resulting
words that are infinite. For this notion of (inverse) image,
Theorem~\ref{thm:full trios} follows with essentially the same proof.

\bibliographystyle{amsplain}

\bibliography{bibliography}

\providecommand{\bysame}{\leavevmode\hbox to3em{\hrulefill}\thinspace}
\providecommand{\MR}{\relax\ifhmode\unskip\space\fi MR }
\providecommand{\MRhref}[2]{%
  \href{http://www.ams.org/mathscinet-getitem?mr=#1}{#2}
}
\providecommand{\href}[2]{#2}
\begin{thebibliography}{10}

\bibitem{DBLP:journals/corr/abs-1809-03108}
Dana Angluin and Dana Fisman, \emph{Regular $\omega$-languages with an
  informative right congruence}, GandALF 2018, Saarbr{\"{u}}cken, Germany,
  2018, 2018, pp.~265--279.

\bibitem{Arnold85b}
Andr{\'e} Arnold, \emph{A syntactic congruence for rational
  $\omega$-languages}, Theoret. Comput. Sci. \textbf{39} (1985), no.~2-3,
  333--335.

\bibitem{BaierBertrandGrosser2012}
Christel Baier, Marcus Gr\"{o}sser, and Nathalie Bertrand, \emph{Probabilistic
  $\omega$--automata}, J. ACM \textbf{59} (2012), no.~1, 1:1--1:52.

\bibitem{Bateman1993}
P.~T. Bateman, C.~G. Jockusch, and A.~R. Woods, \emph{Decidability and
  undecidability of theories with a predicate for the primes}, Journal of
  Symbolic Logic \textbf{58} (1993), no.~2, 672--687.

\bibitem{bojanczykBoundingQuantifier2004}
Miko{\l}aj Boja\'nczyk, \emph{A bounding quantifier}, Computer {{Science
  Logic}}, {{CSL}} 2004, {{Karpacz}} (Jerzy Marcinkowski and Andrzej Tarlecki,
  eds.), Lecture {{Notes}} in {{Computer Science}}, vol. 3210, {Springer},
  2004, pp.~41--55.

\bibitem{bojanczykWeakMSOUnbounding2011}
\bysame, \emph{Weak {{MSO}} with the unbounding quantifier}, Theory Comput.
  Syst. \textbf{48} (2011), no.~3, 554--576.

\bibitem{bojanczykRecognisableLanguagesMonads2015}
\bysame, \emph{Recognisable languages over monads}, CoRR
  \textbf{abs/1502.04898} (2015).

\bibitem{bojanczykUndecidabilityMSOUltimately2018}
Miko{\l}aj Boja\'nczyk, Laure Daviaud, Bruno Guillon, Vincent Penelle, and
  A.~V. Sreejith, \emph{Undecidability of {{MSO}}+"ultimately periodic"},
  Logical Methods in Computer Science \textbf{abs/1807.08506} ((to appear)).

\bibitem{bojanczyk2016}
Miko{\l}aj Boja\'nczyk, Pawel Parys, and Szymon Toru\'nczyk, \emph{The
  {{MSO+U}} theory of $(\mathbb{N}, <)$ is undecidable}, 33rd {{STACS}}
  {{Orl{\'e}ans}}, {{France}} (Nicolas Ollinger and Heribert Vollmer, eds.),
  {{LIPIcs}}, vol.~47, {Schloss Dagstuhl - Leibniz-Zentrum fuer Informatik},
  2016, pp.~21:1--21:8.

\bibitem{bojanczykDeterministicAutomataExtensions2009}
Miko{\l}aj Boja\'nczyk and Szymon Toru\'nczyk, \emph{Deterministic automata and
  extensions of {{Weak MSO}}}, {{FSTTCS}} 2009, {{IIT Kanpur}}, {{India}},
  {{LIPIcs}}, vol.~4, {Schloss Dagstuhl - Leibniz-Zentrum fuer Informatik},
  2009, pp.~73--84.

\bibitem{Buchi60}
J.~Richard B{\"u}chi, \emph{Weak second-order arithmetic and finite automata},
  Z. Math. Logik und Grundl. Math. \textbf{6} (1960), 66--92.

\bibitem{Buchi62}
\bysame, \emph{On a decision method in restricted second order arithmetic},
  Logic, Methodology and Philosophy of Science (Proc. 1960 Internat. Congr .),
  Stanford Univ. Press, Stanford, Calif., 1962, pp.~1--11.

\bibitem{carton2002monadic}
Olivier Carton and Wolfgang Thomas, \emph{The monadic theory of morphic
  infinite words and generalizations}, Information and Computation \textbf{176}
  (2002), no.~1, 51--65.

\bibitem{Eilenberg74}
Samuel Eilenberg, \emph{Automata, languages, and machines. {V}ol. {A}},
  Academic Press [A subsidiary of Harcourt Brace Jovanovich, Publishers], New
  York, 1974.

\bibitem{Elgot61}
Calvin~C. Elgot, \emph{Decision problems of finite automata design and related
  arithmetics}, Trans. Amer. Math. Soc. \textbf{98} (1961), 21--51.

\bibitem{ElgotRabin66}
Calvin~C. Elgot and Michael~O. Rabin, \emph{Decidability and undecidability of
  second (first) order theory of (generalized) successor}, J. of Symbolic Logic
  \textbf{31} (1966), 169--181.

\bibitem{DBLP:journals/jcss/HartmanisH70}
Juris Hartmanis and John~E. Hopcroft, \emph{What makes some language theory
  problems undecidable}, J. Comput. Syst. Sci. \textbf{4} (1970), no.~4,
  368--376.

\bibitem{herre1991hartig}
Heinrich Herre, Micha{\l} Krynicki, Alexandr Pinus, and Jouko
  V{\"a}{\"a}n{\"a}nen, \emph{The h{\"a}rtig quantifier: a survey}, The Journal
  of symbolic logic \textbf{56} (1991), no.~4, 1153--1183.

\bibitem{hummel2012topological}
Szczepan Hummel and Micha{\l} Skrzypczak, \emph{The topological complexity of
  {MSO+U} and related automata models}, Fundamenta Informaticae \textbf{119}
  (2012), no.~1, 87--111.

\bibitem{klaedtkerues2003}
Felix Klaedtke and Harald Rue{\ss}, \emph{Monadic second-order logics with
  cardinalities}, 30th {ICALP} 2003, Eindhoven, The Netherlands, 2003,
  pp.~681--696.

\bibitem{MalerStaiger93}
Oded Maler and Ludwig Staiger, \emph{On syntactic congruences for
  $\omega$-languages}, Theoret. Comput. Sci. \textbf{183} (1997), no.~1,
  93--112.

\bibitem{McNaughton66}
Robert McNaughton, \emph{Testing and generating infinite sequences by a finite
  automaton}, Information and Control \textbf{9} (1966), 521--530.

\bibitem{michalewski2016measure}
Henryk Michalewski and Matteo Mio, \emph{Measure quantifier in monadic second
  order logic}, International Symposium on Logical Foundations of Computer
  Science, Springer, 2016, pp.~267--282.

\bibitem{mioskmi2018}
Matteo Mio, {Micha{\l} Skrzypczak}, and Henryk Michalewski, \emph{Monadic
  second order logic with measure and category quantifiers}, Logical Methods in
  Computer Science \textbf{14} (2018), no.~2.

\bibitem{muchnik03_almos_period_sequen}
Andrei Muchnik, Alexei Semenov, and Maxim Ushakov, \emph{Almost periodic
  sequences}, Theoretical Computer Science \textbf{304} (2003), no.~1-3, 1--33.

\bibitem{nivat}
Maurice Nivat, \emph{Transductions des langages de chomsky}, Annales de
  l'Institut Fourier \textbf{18} (1968), no.~1, 339--455 (fr). \MR{38 \#6909}

\bibitem{DBLP:books/daglib/0016866}
Dominique Perrin and Jean{-}Eric Pin, \emph{Infinite words - automata,
  semigroups, logic and games}, Pure and applied mathematics series, vol. 141,
  Elsevier Morgan Kaufmann, 2004.

\bibitem{Rabin69}
Michael~O. Rabin, \emph{Decidability of second-order theories and automata on
  infinite trees.}, Trans. Amer. Math. Soc. \textbf{141} (1969), 1--35.

\bibitem{RabinovichThomas2006}
Alexander Rabinovich and Wolfgang Thomas, \emph{Decidable theories of the
  ordering of natural numbers with unary predicates}, Computer Science Logic
  (Berlin, Heidelberg) (Zolt{\'a}n {\'E}sik, ed.), Springer Berlin Heidelberg,
  2006, pp.~562--574.

\bibitem{Robinson1958}
Raphael~M. Robinson, \emph{Restricted set-theoretical definitions in
  arithmetic}, Proceedings of the American Mathematical Society \textbf{9}
  (1958), no.~2, 238--238.

\bibitem{Semenov84b}
Alexei~L. Semenov, \emph{Decidability of monadic theories}, Mathematical
  foundations of computer science, 1984 (Prague, 1984), Springer, Berlin, 1984,
  pp.~162--175.

\bibitem{siefkes1971undecidable}
Dirk Siefkes, \emph{Undecidable extensions of monadic second order successor
  arithmetic}, Mathematical Logic Quarterly \textbf{17} (1971), no.~1,
  385--394.

\bibitem{thomas1975note}
Wolfgang Thomas, \emph{A note on undecidable extensions of monadic second order
  successor arithmetic}, Archive for Mathematical Logic \textbf{17} (1975),
  no.~1, 43--44.

\bibitem{thomasLanguagesAutomataLogic1997}
\bysame, \emph{Languages, automata, and logic}, Handbook of {{Formal
  Languages}}, {{Volume}} 3: {{Beyond Words}}. (Grzegorz Rozenberg and Arto
  Salomaa, eds.), {Springer}, 1997, pp.~389--455.

\bibitem{trakthenbrot}
Boris~A. Trakhtenbrot, \emph{Finite automata and the logic of single-place
  predicates}, Dokl. Akad. Nauk SSSR \textbf{140} (1961), no.~2, 326--329.

\bibitem{zetzsche16_boolean_closed_full_trios_ration_kripk_frames}
Georg Zetzsche, Dietrich Kuske, and Markus Lohrey, \emph{On {Boolean} closed
  full trios and rational {Kripke} frames}, Theory of Computing Systems
  \textbf{60} (2017), no.~3, 438--472.

\end{thebibliography}



\end{document}